\documentclass{article}

\usepackage{graphicx}      
\usepackage{natbib}        

\usepackage{graphicx}
\usepackage{textcomp}
\usepackage{algorithm}
\usepackage{algorithmic}
\usepackage{xcolor}
\usepackage{mathtools}
\usepackage{xfrac}
\usepackage{comment}
\usepackage{pgfplots}
\pgfplotsset{compat=1.18}
\usepackage{pgfmath}
\usepackage{pgfplotstable}
\usetikzlibrary{calc,decorations.markings,intersections}
\usepgfplotslibrary{fillbetween,patchplots}
\usepackage{amsmath,amssymb,amsfonts}
\newcommand\scalemath[2]{\scalebox{#1}{\mbox{\ensuremath{\displaystyle #2}}}}
\usepackage{enumitem}

\DeclareMathOperator*{\emm}{\text{\fontencoding{U}\fontfamily{boondoxuprscr}\fontshape{n}\selectfont m}}

\newcommand{\f}{\mathsf{f}}
\newcommand{\ed}{\mathsf{e}}
\newcommand{\ve}{\mathsf{v}}
\newcommand{\R}{\mathbb{R}}

\newcommand{\defeq}{:=}
\newcommand{\tr}{\intercal}
\newcommand{\Poly}{{\mathcal P}}
\newcommand{\cvx}{{\mathrm{convh}}}

\newcommand{\vp}{{\mathrm{v}}}

\definecolor{sampathcolor}{HTML}{ff6f69}

\newcommand{\norm}[1]{\left\lVert#1\right\rVert}

\DeclareFontFamily{U}{boondoxuprscr}{\skewchar \font =45}
\DeclareFontShape{U}{boondoxuprscr}{m}{n}{
    <-> BOONDOXUprScr-Regular}{}
\DeclareFontShape{U}{boondoxuprscr}{b}{n}{
    <-> BOONDOXUprScr-Bold}{}

\newtheorem{theorem}{Theorem}

\newtheorem{lemma}{Lemma}
\newtheorem{remark}{Remark}
\newtheorem{assumption}{Assumption}
\newtheorem{proposition}{Proposition}

\newtheorem{definition}{Definition}
\newtheorem{pf}{Proof}

\title{On Piecewise Quadratic Terminal Costs for MPC}

\author{
	Sampath Kumar Mulagaleti\thanks{IMT School for Advanced Studies Lucca, Italy}
	\and
	Boris Houska\thanks{ShanghaiTech University, School of Information Science and Technology, China}
	\and
	Mario Zanon\footnotemark[1]
	\and
	Mario E.~Villanueva\footnotemark[1]
}
\begin{document}
	\maketitle
	

\begin{abstract}                
This paper presents a novel approach to synthesize stabilizing terminal ingredients for linear model predictive control (MPC) schemes, with the aim of increasing the region of attraction while reducing suboptimality with respect to the solution of the infinite-horizon optimal control problem. It is based on the construction of a novel terminal region using methods from the field of configuration-constrained polytopic computing, along with a terminal cost that is exactly equal to the infinite-horizon linear-quadratic regulator cost in a nontrivial neighborhood of the steady-state.
The practical performance of the controller is illustrated through various case studies, and comparisons with state-of-the-art approaches are presented.
\end{abstract}


\section{Introduction}

Model Predictive Control (MPC) is a widely used advanced control strategy that optimizes a sequence of control actions over a finite horizon~\citep{Rawlings2009}. At each sampling instant, MPC solves an optimization problem based on the current state, accounting for dynamics, constraints, and objectives over the prediction horizon. To guarantee recursive feasibility and stability, terminal conditions are typically required~\citep{Chen1998}; however, they are often omitted in practice~\citep{Qin2003}, with long horizons used instead~\citep{Gruene2009}. For systems with coupled slow and fast dynamic modes such as those encountered in process control, such horizons may become impractically long, making online optimization difficult. In such cases, it becomes important to design suitable terminal ingredients accounting for long-term behavior.

Terminal constraints for MPC were introduced in~\citet{Chen1998}, and substantial progress has since been made on terminal ingredients and their properties~\citep{Rawlings2009}. A key requirement for terminal regions is control invariance to ensure recursive feasibility and methods to compute such sets are well developed in set-theoretic control~\citep{Blanchini2008}. The terminal region shapes the MPC admissible set, and an ideal choice makes it approximate the maximal control invariant set as tightly as possible with a short prediction horizon.
The terminal cost is equally important: it should be a control Lyapunov function (CLF) for the stage cost. CLFs were introduced in~\citet{Zubov1965}, with foundational developments in~\citet{Artstein1983} and an overview in~\citet{Giesl2015}. The terminal cost influences MPC closed-loop behavior, and an ideal choice bounds the infinite-horizon cost as tightly as possible.

%
In linear-quadratic MPC, a standard choice of terminal ingredients uses the infinite-horizon unconstrained OCP to obtain a linear feedback law with a quadratic value function. The terminal region is then set as the maximal admissible control invariant set for this feedback, and the value-function defines a quadratic terminal cost. While this ensures stability, it can yield overly small admissible regions~\citep{johansson2024}. This motivates using more general invariant terminal regions, raising the problem of designing CLFs that can serve as terminal cost. For gauge-based stage costs, CLFs are synthesized in~\citep{Rakovic2012_gauge,Darup2015}. With quadratic stage costs,~\citep{Grammatico2013} proposed a polyhedral CLF over a contractive set, and~\citep{johansson2024} similarly exploits piecewise linear invariance to construct a CLF parameterized as a common quadratic Lyapunov function.



\indent
\textit{Contributions:} 
We revisit the computation of invariant terminal regions and piecewise quadratic terminal costs for MPC using recent advances in polytopic computing~\citep{Villanueva2024}. Leveraging configuration-constrained (cc-)polytopes, which admit a joint facet-vertex representation, we derive results enabling novel designs of terminal ingredients for linear-quadratic MPC. The main contributions of this paper are:
\begin{enumerate}
\item A new class of control invariant sets containing, under suitable assumptions, the union of all $\beta$-contractive cc-polytopes; see Lemma~\ref{lem:T_beta}.

\item A new class of piecewise quadratic CLFs, defined on the above terminal region; see Theorem~\ref{thm::clf}.
\end{enumerate}

Numerical examples show that the resulting MPC scheme outperforms \citep{Grammatico2013} and \citep{johansson2024} in admissible-region size and closed-loop optimality, and it explicitly characterizes a nontrivial region where it reduces to the linear control law when constraints are inactive.

\indent
\textit{Outline:} Section~\ref{sec::mpc} introduces the MPC problem and reviews control invariant sets and CLFs. Section~\ref{sec:polytopic_regions} recalls cc-polytopes, proposes a new control invariant set, and analyses its properties. Section~\ref{sec::Cost} defines a piecewise quadratic CLF over this set and studies its stability properties. Section~\ref{sec:complete_formulation} presents the MPC formulation, establishes recursive feasibility and stability under the proposed terminal ingredients, and analyzes computational complexity. Finally, Section~\ref{sec:numerics} presents two numerical examples and compares the proposed method with benchmark approaches.

\textit{Notation:} For a given matrix $Q \in \mathbb{R}^{n \times n}$, we use the shorthand $\| x \|_Q^2 := x^\tr Q x$. A discussion of conditions on the matrix $F$ under which $\Poly(y) := \{ x \in \mathbb{R}^n \mid F x \leq y \}$ is a polytope is found in~\citep[Section~2.3]{Houska2025}. Given compact convex sets $Z_1,Z_2 \subseteq \R^{n_x}$ with $Z_1 \subseteq Z_2$,  $d(Z_1;Z_2):=\min\{\epsilon \geq 0 : Z_2 \subseteq Z_1 \oplus \epsilon \mathcal{B}^{n_x}_{\infty}\}$ denotes their Hausdorff distance, with $\mathcal{B}^{n_x}_{\infty}$ the $\infty$-norm ball in $\R^{n_x}$. 

\section{Model Predictive Control}
\label{sec::mpc}
This paper concerns linear-quadratic infinite-horizon optimal control problems of the form
\begin{equation}
\label{eq::ocp}
\hspace{-5pt}
V_\infty(x_0) = \min_{x,u} \sum_{k=0}^\infty \ell(x_k,u_k)  \ \text{s.t.} 
\left\{
\scalemath{0.95}{\begin{array}{l}
k\in\mathbb{N},\\
x_{k+1} = Ax_{k} + Bu_{k},\\
x_k \in \mathbb X, \ u_k \in \mathbb U.
\end{array}}
\right.
\end{equation}
Here, $\mathbb X \subseteq \mathbb{R}^{n_x}$ and $\mathbb U \subseteq \mathbb{R}^{n_u}$ denote state- and control constraint sets, $A$ and $B$ are given system matrices, and
\begin{align}
\label{eq::stage_cost}
\ell(x,u)= x^\tr Q x + 2 x^\tr S u + u^\tr R u
\end{align}
is a given quadratic stage cost. We denote our optimization variables consisting of the state and control sequences as $x = (x_1,x_2,\ldots)$ and $u = (u_0,u_1,\ldots)$, with $x_0$ a given initial state vector that is not optimized. We work with the following assumption; see~\citep{Rawlings2017}.

\begin{assumption}
\label{ass::sets}
The sets $\mathbb X$ and $\mathbb U$ are closed, convex and we have $0 \in \mathbb X$ as well as $0 \in \mathbb U$.
\end{assumption}

\begin{assumption}
\label{ass::weights}
The matrix $\left(\begin{array}{cc}
Q & S \\
S^\tr & R
\end{array}\right)$
is positive semi-definite, $R$ is positive definite, and $(A,Q)$ is observable~\citep{Kalman1960,Rawlings2017}.
\end{assumption}

\subsection{Linear Quadratic Regulator}
Throughout this article we make use of the algebraic Riccati equation~\citep{Kalman1960}, which has the form
\begin{align}
P &= A^\tr P A + Q \label{eq::DARE} \\
& \qquad - ( A^\tr P B + S ) (B^\tr P B + R)^{-1} ( A^\tr P B + S )^\tr \notag \\
0 &= (R+B^{\tr}PB)K + (A^{\tr}PB + S)^\tr. \nonumber
\end{align}

\begin{assumption}
\label{ass::LQR}
We assume that~\eqref{eq::DARE} admits a solution $(P,K)$ with $P$ being symmetric and positive definite.
\end{assumption}

\noindent
If Assumption~\ref{ass::weights} holds, $(A,B)$ is stabilizable if and only if Assumption~\ref{ass::LQR} is satisfied~\citep{Kalman1960,Zanon2022}. Moreover, if $\mathbb X = \mathbb R^{n_x}$ and $\mathbb U = \mathbb R^{n_u}$, the optimal value of~\eqref{eq::ocp} is $V_\infty(x_0)=x_0^\tr P x_0$ with 
\begin{align}
    \label{eq:LQR_CL}
    \forall k\in\mathbb{N}, \quad u_k=Kx_k, \quad x_{k+1} = (A+BK)x_{k}.
\end{align}
The corresponding LTI control law, $\mu_\infty(x) = Kx$, is known as the linear-quadratic regulator (LQR).

We work with the following definition of control invariance.
\begin{definition}
A set $X \subseteq \mathbb X$ is called control invariant if
\begin{align}
\label{eq::CIdef}
\forall x \in X, \exists u \in \mathbb U: \quad Ax+Bu \in X.
\end{align}
\end{definition}
A set $\mathbb{T}_{\sf LQR} \subseteq \mathbb X$ is control invariant for $u = Kx$ if
\begin{align}
\label{eq::TLQR}
(A+BK)\mathbb{T}_{\sf LQR} \subseteq \mathbb{T}_{\sf LQR} \quad \text{and} \quad K\mathbb{T}_{\sf LQR}\subseteq \mathbb{U}.
\end{align}
If~\eqref{eq::TLQR} holds, \eqref{eq::CIdef} is satisfied with $u = Kx, X = \mathbb T_{\sf LQR}$. If Assumptions~\ref{ass::sets},~\ref{ass::weights}, and~\ref{ass::LQR} hold, $\mathbb{T}_{\sf LQR} = \{ 0 \}$ satisfies~\eqref{eq::TLQR}. Every candidate $\mathbb{T}_{\mathrm{LQR}}$ satisfies  $0 \in \mathbb{T}_{\sf LQR}$ since Assumption~\ref{ass::LQR} ensures $\scalemath{0.95}{\{ 0 \} = \underset{{k \to \infty}} {\lim}(A+BK)^k \cdot \mathbb{T}_{\sf LQR} \ \subseteq \ \mathbb{T}_{\sf LQR}}$.
Methods to compute $\mathbb{T}_{\sf LQR}$ can be found in~\citep{Houska2025}.

\subsection{Model Predictive Control}
\label{sec:refer_back}
Since an explicit solution of the infinite horizon OCP~\eqref{eq::ocp} is in general not available, the main idea of MPC is to solve
\begin{equation}\label{eq:mpc}
    \begin{alignedat}{2}
        V_N(x_0) \ = \ &\!\min_{x,u} \ &&\sum^{N-1}_{k=0} \ell(x_k,u_k) + \emm(x_N) \\
        & \ \text{s.t.}  &&\left\{ 
            \begin{aligned}
                &\forall k\in\{0,1,\ldots,N-1\},\\
                & x_{k+1} = Ax_{k} + Bu_{k},\\
                & x_{k}\in\mathbb{X}, \ u_{k}\in\mathbb{U}
            \end{aligned}
        \right.
    \end{alignedat}
\end{equation}
instead. Here, $x = (x_1,\ldots,x_N)$ and $u = (u_0,\ldots,u_{N-1})$ are optimization variables, and $N$ the prediction horizon. The terminal cost $\emm: \mathbb R^{n_x} \to \mathbb R \cup \{ \infty \}$ is assumed a convex positive definite CLF
satisfying the Lyapunov condition
\begin{align}
\label{eq::Lyapunov}
\hspace{-5pt}
\scalemath{0.965}{\forall x \in \mathbb R^{n_x}, \min_{u \in \mathbb U} \left( \ell(x,u) + I_{\mathbb X}(x) + \emm(Ax+Bu) \right) \leq \emm(x),}    
\end{align}
where $I_{\mathbb X}$ is the indicator function of the set $\mathbb{X}$.
If Assumptions~\ref{ass::sets} and~\ref{ass::weights} hold, \eqref{eq::Lyapunov} ensures that $V_N \geq V_\infty$ is a CLF~\citep{Rawlings2017}. Consequently,~\eqref{eq:mpc} defines a strictly stabilizing and recursively feasible control law $\mu_N(x_0) = u_0^\star(x_0)$, where $u_0^\star(\cdot)$ denotes the first element of the optimal control inputs of~\eqref{eq:mpc}. If Assumptions~\ref{ass::sets},~\ref{ass::weights}, and~\ref{ass::LQR} hold and $\mathbb T_{\sf LQR}$ satisfies~\eqref{eq::TLQR}, then the function
\begin{align}
\label{eq::emmLQR}
\emm(x) = \left\{
\begin{array}{ll}
x^\tr P x & \text{if} \ x \in \mathbb T_\mathrm{LQR} \\
\infty & \text{otherwise}
\end{array}
\right.
\end{align}
satisfies~\eqref{eq::Lyapunov}. 
The set $\{ x\in\mathbb{R}^
{n_x} \, | \, \emm(x) < \infty \}$ is the terminal region, and $\mathcal{O}_N =
\{ x\in\mathbb{R}^
{n_x} \, | \, V_{N}(x) < \infty \}$ is the admissible region. The nature of $\mathcal{O}_N$ depends on the terminal region and $N$.  
In this paper, we propose a new CLF that satisfies the following properties:
\begin{enumerate}[label=(\alph*)]
    \item It results in a larger admissible set than that of Problem \eqref{eq:mpc} when formulated with terminal cost \eqref{eq::emmLQR}. 
    \item The formulation \eqref{eq:mpc} constitutes a convex optimization problem; and
    \item There exists a non-singleton set containing the origin where $\mu_N(x_0) = Kx_0$, i.e., the control law reduces to the LQR control law when constraints are not active.
\end{enumerate}
{\color{black} \textbf{Motivating example.} To motivate the development of terminal ingredients that increase the domain of attraction of MPC schemes, consider a simple scalar system with dynamics $x^+=ax+u$, where $a>0$, $u \in [-\bar{u},\bar{u}]$, and no state constraints. Defining the stage cost as $\ell(x,u)=x^2+u^2$, the stabilizing feedback gain satisfying \eqref{eq::DARE} can be derived as $K(a) = -a(a^2+\sqrt{a^4+4})/(2+a^2+\sqrt{a^4+4})$. The maximal positive invariant set under the LQR law $u=K(a)x$ is then given by $\mathbb{T}_{\mathrm{LQR}}=[-b(a),b(a)]$, where $b(a)=\bar{u}(2+a^2+\sqrt{a^4+4})/(a(a^2+\sqrt{a^4+4}))$. Observe that $b'(a)<0$ for all $a>0$, implying that the set $\mathbb{T}_{\mathrm{LQR}}$ shrinks as $a$ increases. In process control applications, \textit{slow} dynamical systems are frequently encountered \citep{seborg2010process}; these correspond to values of $a \approxeq 1$, and sometimes $a>1$ (unstable). For such systems, terminal sets constructed as in \eqref{eq::emmLQR} can be exceedingly small. Then, obtaining a reasonably large admissible region $\mathcal{O}_N$ may require an excessively long horizon $N$. Recall that $\mathcal{O}_N$, the $N$-step backward reachable set from the terminal set, is given by $\mathcal{O}_N=[-c_N(a),c_N(a)]$, where $c_N(a)=(b(a)+\sum_{i=0}^{N-1}a^i\bar{u})/a^N$. It can be verified that $c_{N+1}(a)-c_N(a)=(\bar{u}-(a-1)b(a))/a^{N+1}>0$. Because the per-step growth of the admissible set decays proportionally to $a^{-N}$, increasing the horizon length yields diminishing returns for $a \gtrsim 1$. This further motivates enlarging the terminal set itself.}

\section{Polytopic Terminal Regions}
\label{sec:polytopic_regions}
This section revisits recent ideas from~\citep{Villanueva2024} on configuration-constrained polytopic computing and elaborates on how these developments can be used to derive new classes of terminal regions for MPC.

\subsection{Configuration-Constrained Polytopes}
The idea of many polyhedral modeling and computing methods~\citep{Houska2025} is to select a facet matrix $F\in\mathbb{R}^{{\sf f} \times n_x}$ and consider parametric sets of the form
\begin{equation}
\label{eq:template_polytope}
\Poly(y) = \left\{ x\in\mathbb{R}^{n_x} \middle| 
Fx\leq y \right\}.
\end{equation}
Under the assumption that $\Poly(0)=\{0\}$, $\mathcal{P}(y)$ is a (potentially empty) polytope for every $y \in \R^{\sf f}$~\citep{Ziegler1995}.
In general, however, the number of vertices of $\mathcal P(y)$ depends on the choice of $y$ and $F$,
which hinders the development of polytopic computing methods that need access to both facet and vertex representations. Nevertheless, as shown in~\citep[Theorem~2]{Villanueva2024}, one can determine an edge matrix $E\in\mathbb{R}^{{\sf e}\times{\sf f}}$ and a collection of vertex matrices $V = (V_1,\ldots,V_{\sf v})$, with $V_i \in \mathbb{R}^{n_x\times{\sf f}}$, such that $\mathcal P(y)$ is equal to the convex hull of the points $V_1 y, V_2 y, \ldots V_{\sf v} y$ if and only if the parameter $y \in \mathbb R^{\sf y}$ satisfies $E y \leq 0$. As the construction of such triples $(F,E,V)$ is already discussed in~\citep[Section~3.5]{Villanueva2024} and in even more detail in~\citep[Sections~2.12 and~2.13]{Houska2025}, we directly work with the following assumption.
\begin{assumption}
\label{ass::FEV}
The triple $\scalemath{0.9}{(F,E,V)}$ satisfies
\begin{equation}
\label{eq:cc_relation}
\mathcal P(y) = \cvx(\{V_1y,\ldots,V_{\sf v}y\}) \quad \Longleftrightarrow \quad Ey\leq 0,
\end{equation}
where $\cvx$ denotes the convex hull.
\end{assumption}
\noindent
As proposed in~\citep{Villanueva2024}, we call the constraint $E y \leq 0$ a configuration constraint, and refer to the polytopes $\Poly(y)$ that comply with this constraint as configuration-constrained polytopes (cc-polytopes).

We now characterize vectors $y \in \mathbb{R}^{\mathsf{f}}$ satisfying $Ey \leq 0$ such that $X = \mathcal{P}(y)$ is a control invariant set, that is, such that it fulfills \eqref{eq::CIdef}. This characterization relies on the set
\begin{align}
    \mathbb{S} := \left\{ (y, v, y^+) \;\middle|\; 
    \begin{array}{l}
        \forall i \in \{1, \ldots, \mathsf{v}\}, \\[5pt]
        F(AV_i y + B v_i) \leq y^+, \\[5pt]
        Ey \leq 0, \ V_i y \in \mathbb{X}, \ v_i \in \mathbb{U}
    \end{array}
    \right\},
\end{align}
where we introduce the shorthand $v = (v_1, \ldots, v_{\mathsf{v}}) \in \mathbb{R}^{\mathsf{v} n_u}$.
\begin{proposition}
\label{prop::FCI}
Let Assumptions~\ref{ass::sets} and~\ref{ass::FEV} hold. Then, there exists $(y, v, y^+) \in \mathbb{S}$ if and only if for every $x \in \Poly(y)$, there exists a control $u \in \mathbb{U}$ such that $Ax + Bu \in \Poly(y^+)$.
\end{proposition}
\begin{pf}
The statement follows from~\citep[Corollary~4]{Villanueva2024}; see also~\citep[Sect.~3]{Houska2025}.
\end{pf}
Due to Proposition \ref{prop::FCI}, for every $x \in \mathcal{P}(y)$, there exists a nonnegative vector $\lambda \in \R^{\sf v}$ satisfying $x = \sum_{i=1}^{\vp} \lambda_i V_i y$ and $\norm{\lambda}_1=1$ such that $Ax+Bu \in \mathcal{P}(y^+)$ holds with $u= \sum_{i=1}^{\vp} \lambda_i v_i \in \mathbb{U}$. Next, it follows from the vertex control theorem \citep{Gutman1986} that $\Poly(y_\mathrm{s})$ is a control invariant cc-polytope if and only if there exists $v \in \R^{\ve n_u}$ satisfying $(y_\mathrm{s},v,y_\mathrm{s}) \in \mathbb{S}$. We parameterize the vertex control inputs as $v_i=KV_iy$ for $i \in \{1,\cdots,\ve\}$, where $K$ denotes the LQR feedback gain, and define
\begin{align*}
    \mathbb Y_{\sf LQR}:=\left\{ y_{\sf s} \in \R^{\f} \ \middle|  \begin{array}{l}  \forall \ i \in \{1,\cdots,\ve\}, \vspace{5pt} \\
    F(A+BK) V_i y_{\sf s} \leq y_{\sf s}, \vspace{5pt}\\ 
    Ey \leq 0, \ V_i y_{\sf s} \in \mathbb{X}, \ KV_i y_{\sf s} \in \mathbb{U}
    \end{array} \right\}.
\end{align*}
This ensures that $\mathbb{T}_{\mathrm{LQR}} = \mathcal{P}(y_{\sf s})$ satisfies~\eqref{eq::emmLQR} for any $y_{\sf s} \in \mathbb{Y}_{\sf LQR}$, such that $\mathcal{P}(y_{\sf s})$ is an invariant cc-polytope with $u = Kx$.
A straightforward implementation of terminal regions for MPC using cc-polytopes involves selecting any $y_{\sf s} \in \mathbb{Y}_{\sf LQR}$ offline and setting $\mathbb{T}_{\sf LQR} = \Poly(y_{\sf s})$ in~\eqref{eq::emmLQR}. Alternatively, a  $\mathbb{T}_{\sf LQR}$ can be computed online via
\begin{align}
\label{eq::emmLQRopt}
\emm(x) = \min_{y_{\sf s}} \ x^\top P x \quad \text{s.t.} \quad 
y_{\sf s} \in \mathbb{Y}_{\sf LQR}, \quad
F x \leq y_{\sf s},
\end{align}
since any feasible $y_{\sf s}$ in~\eqref{eq::emmLQRopt} defines a valid terminal region, i.e., $X=\mathcal{P}(y_{\sf s})$ satisfies \eqref{eq::CIdef}. While this introduces an additional optimization variable $y_{\sf s}$, its advantage lies in having a domain at least as large as that of~\eqref{eq::emmLQR} for a pre-computed cc-polytopic terminal region. Note that~\eqref{eq::emmLQRopt} is equivalent to defining the terminal region as
\begin{align}
\label{eq::hatTLQR}
\widehat{\mathbb{T}}_{\sf LQR} = \{ x \in \mathbb{R}^{n_x} \mid \exists y_{\sf s} \in \mathbb{Y}_{\sf LQR}: Fx \leq y_{\sf s} \}.
\end{align}
The set $\widehat{\mathbb{T}}_{\sf LQR}$ is generally not a cc-polytope. Instead, given $(F,E,V)$, it is the union of all cc-polytopes satisfying~\eqref{eq::emmLQR}.
\subsection{Novel Representation of a Control Invariant Set}
We introduce a novel control invariant set, used to define the terminal cost $\emm$ in the next section. We define it as
\begin{align}
\label{eq::Tbeta}
\mathbb{T}(\beta) := \left\{ x \in \R^{n_x} \middle| \begin{array}{l} \exists y,y_{\sf s} \in \R^\f, \exists v \in \R^{\ve n_u}:\\
Fx \leq y, \ y_{\sf s} \in \mathbb Y_{\sf LQR}, \\ 
(y,v,y_{\sf s}+\beta(y-y_{\sf s})) \in \mathbb{S}
\end{array}
\right\}.
\end{align}
\begin{lemma}
\label{lem:T_beta}
Let Assumptions~\ref{ass::sets},~\ref{ass::weights},~\ref{ass::LQR}, and~\ref{ass::FEV} be satisfied. Then the following statements hold for all $\beta \in [0,1]$:
\begin{enumerate}
\item The set $\mathbb T(\beta)$ is closed, convex, and control invariant.
\item If $\mathbb X$ and $\mathbb U$ are polyhedra, then $\mathbb T(\beta)$ is a polyhedron.
\item We have $\mathbb T(\beta) \supseteq \widehat{\mathbb T}_{\sf LQR}$, where $\widehat{\mathbb T}_{\sf LQR}$ is defined in~\eqref{eq::hatTLQR}.
\end{enumerate}
\end{lemma}
\begin{pf} The fact that $\mathbb T(\beta)$ is closed and convex follows from Assumption~\ref{ass::sets} and~\eqref{eq::Tbeta}. Next, let $x \in \mathbb T(\beta)$ such that there exist $y \in \R^\f$, $y_{\sf s} \in \mathbb Y_{\sf LQR}$, and $v \in \R^{\vp n_u}$ with $x \in \Poly(y)$ and $(y,v,y^+) \in \mathbb S,$ where $y^+:=y_{\sf s} + \beta (y-y_{\sf s})$. Under Assumptions~\ref{ass::sets} and~\ref{ass::FEV}, Proposition~\ref{prop::FCI} guarantees there exists some $u \in \mathbb U$ such that $Ax + Bu \in \Poly(y^+)$.
Thus, if $\Poly(y^+) \subseteq \mathbb T(\beta)$ holds, then $\mathbb T(\beta)$ is control invariant. To this end, we introduce the auxiliary vertex inputs
\begin{align}
\label{eq::vs}
v_{\sf s} := \left( K V_1 y_{\sf s}, \ldots, K V_{\vp} y_{\sf s} \right),
\end{align}
satisfying $(y_{\sf s},v_{\sf s},y_{\sf s}) \in \mathbb S$ as $y_{\sf s} \in \mathbb Y_{\sf LQR}$. Thus, defining $v^+:=(1-\beta) v + \beta v_{\sf s}$ and $y^{++}:= (1-\beta)y^+ + \beta y_{\sf s}$, the convexity of $\mathbb{S}$ implies $(1-\beta)(y,v,y^+)+\beta (y_{\sf s},v_{\sf s},y_{\sf s}) \in \mathbb S,$ or, equivalently, $(y^+,v^+,y^{++}) \in \mathbb{S}$ for any $\beta \in [0,1]$. 
This implies that the triple $(y^+,v^+,y_{\sf s})$ satisfies all conditions in the definition of $\mathbb T(\beta)$, such that the inclusion $\Poly(y^+) \subseteq \mathbb T(\beta)$ follows. Hence, $\mathbb{T}(\beta)$ is a control invariant set.
The second statement follows directly from Assumption~\ref{ass::sets} and the construction in~\eqref{eq::Tbeta}. The third statement follows upon setting $y = y_{\sf s}$ and $v = v_{\sf s}$ in the definition of $\mathbb T(\beta)$ in~\eqref{eq::Tbeta}.
\end{pf}
\begin{remark}
\label{remark:include_bc}
Since we may set $y_{\sf s}=0$ in $\mathbb{T}(\beta)$, it includes all cc-polytopes satisfying the $\beta$-contractivity condition $(y,v,\beta y) \in \mathbb{S}$. Hence, $\mathbb{T}(\beta)$ is a superset of the union of all $\beta$-contractive cc-polytopes, such that the terminal region is larger than the maximal $\beta$-contractive cc-polytope.
\end{remark}
{\color{black}
\textbf{Motivating example \textit{(continued)}.} Returning to the example in Sec.~\ref{sec:refer_back}, we now demonstrate how employing $\mathbb{T}(\beta)$ as the terminal set enlarges the MPC domain of attraction. Let us define $\mathbb{T}(\beta)$ with $F=[1 \ -1]^{\top}$, such that $\sf{f}=\ve=2$ and $E=[-1 \ -1]$. Because the maximal control invariant set $\mathbb{T}_{\mathrm{LQR}}=[-b(a),b(a)]$ under $u=K(a)x$ is defined by the same hyperplanes, it follows that $\hat{\mathbb{T}}_{\mathrm{LQR}}=\mathbb{T}_{\mathrm{LQR}}$. By Lemma \ref{lem:T_beta}-(3), we have $\mathbb{T}_{\mathrm{LQR}} \subseteq \mathbb{T}(\beta)$. We now derive a lower bound on $\beta$ to ensure this inclusion holds strictly. Consider a $\beta$-contractive set $\mathbb{C}_{\beta}=[-m(a),m(a)]$, which satisfies $\forall x \in \mathbb{C}_{\beta}, \ \exists u \in [-\bar{u},\bar{u}] : ax+u \in \mathbb{C}_{\beta}$. For any $\beta<\min\{a,1\}$, $\mathbb{C}_{\beta}$ defined with $m(a)=\bar{u}/(a-\beta)$ is the maximal $\beta$-contractive set. From Remark \ref{remark:include_bc}, we know that $\mathbb{C}_{\beta} \subseteq \mathbb{T}(\beta)$. Consequently, selecting $\beta> 2a/(2+a^2+\sqrt{a^4+4})$ ensures $b(a) < m(a)$, or equivalently $\mathbb{T}_{\mathrm{LQR}} \subset \mathbb{C}_{\beta}$, meaning $\mathbb{T}_{\mathrm{LQR}} \subset \mathbb{T}(\beta)$ strictly holds. This lower bound on $\beta$ increases for $a \in [0,\sqrt{2}]$ and decreases for $a> \sqrt{2}$, reaching a maximum value of $1/(1+\sqrt{2})<1$. Note that if $a<1$ and $\beta \in [a,1]$, then $\mathbb{C}_{\beta}=\mathbb{R}$, satisfying $\mathbb{T}_{\mathrm{LQR}} \subset \mathbb{T}(\beta)$ by construction. Thus, for this scalar system, we guarantee a larger MPC admissible region for a given horizon length. Specifically, the admissible set $\mathcal{O}_N=[-c_N(a),c_N(a)]$ is bounded by $c_N(a)\geq (m(a)+\sum_{i=0}^{N-1}a^i\bar{u})/a^N$. While the preceding analysis considered a scalar system, the insight extends to multidimensional stiff systems. For a diagonalizable system with $A = T \operatorname{diag}(\lambda_1, \ldots, \lambda_{n_x}) T^{-1}$, the extent of $\mathbb{T}_{\mathrm{LQR}}$ along the $i$-th modal direction is governed by $|\lambda_i|$ and the effective control authority in that direction (equivalent of $\bar{u}$).
The mode with the largest $|\lambda_i|$ therefore constitutes the bottleneck for enlarging the admissible region via horizon extension, and the benefits of the enlarged terminal set $\mathbb{T}(\beta)$ apply directly to this limiting mode.
}


\section{Terminal Cost Functions}
\label{sec::Cost}
The main idea of this section is to introduce a terminal cost function $\emm:\mathbb{R}^{n_x}\to\mathbb{R}\cup\{\infty\}$, which is defined as
\begin{align}
\emm(x) := &\min_{y,v,y_{\sf s}} \ \Vert x \Vert_{P}^{2} + \Vert y-y_{\sf s} \Vert_{\Gamma}^2 + \sum^{\sf v}_{i=1}\Vert v_i - KV_iy \Vert_{\Theta}^2  \nonumber \\
& \ \ \  \text{s.t.} \
\begin{cases}
Fx \leq y, \ \ y_{\sf s} \in \mathbb Y_{\sf LQR},   \\
(y,v,y_{\sf s}+\beta(y-y_{\sf s})) \in \mathbb{S},
\end{cases} \label{eq::EMM}
\end{align}
where the matrices $\Gamma \in\mathbb{R}^{\f \times \f}$ and $\Theta \in\mathbb{R}^{n_{u}\times n_{u}}$ are positive semi-definite, such that~\eqref{eq::EMM} is a convex problem.
\begin{proposition}
\label{prop::emm}
Let Assumptions~\ref{ass::sets},~\ref{ass::weights},~\ref{ass::LQR}, and~\ref{ass::FEV} be satisfied, $\beta \in [0,1]$, and let $\Gamma$ and $\Theta$ be positive semi-definite. Then:
\begin{enumerate}
\item The function $\emm$ is a convex function.
\item We have $\emm(x)<\infty$ if and only if $x \in \mathbb T(\beta)$.
\item If $\mathbb X$ and $\mathbb U$ are polyhedra, then $\emm$ is a continuous and piecewise quadratic function on $\mathbb T(\beta)$.
\item We have $\emm(x) = x^\tr P x$ for all $x \in \widehat{\mathbb T}_{\sf LQR}$.
\end{enumerate}
\end{proposition}
\begin{pf}
    See Appendix.
\end{pf}
 Proposition~\eqref{prop::emm} ensures that $\emm$ coincides with the optimal infinite horizon value function $V_\infty$ on $\widehat{\mathbb T}_{\sf LQR}$. Hence, it satisfies the Lyapunov descent condition~\eqref{eq::Lyapunov} on $\widehat{\mathbb T}_{\sf LQR}$. Next, we provide a simple condition on the parameter $\beta$ and the weight matrix $\Theta$ under which it satisfies~\eqref{eq::Lyapunov} on the much larger domain $\mathbb{T}(\beta)$. 

We now present a condition on the weight matrix $\Theta$ under which the function $\emm$ from~\eqref{eq::EMM} satisfies the Lyapunov condition~\eqref{eq::Lyapunov} for any given $\beta \in [0,1)$. 
\begin{theorem}\label{thm::clf}
Let Assumptions~\ref{ass::sets},~\ref{ass::weights},~\ref{ass::LQR}, and~\ref{ass::FEV} be satisfied, and assume that $\beta \in [0,1)$ and that the weight matrix $\Gamma$ is positive semi-definite. If the weight matrix $\Theta$ satisfies
\begin{align}
\label{eq::Tcond}
\Theta \succeq \frac{B^{\tr}PB + R}{1-\beta^2},
\end{align}
then the function $\emm$ in~\eqref{eq::EMM} satisfies the Lyapunov condition~\eqref{eq::Lyapunov} for the stage cost $\ell(x,u) = x^\tr Q x + 2x^\tr S u + u^\tr R u$.
\end{theorem}

\begin{pf}
For any $x \notin \mathbb{X}$, we have that $x \notin \mathbb{T}(\beta)$, such that $\emm(x)=\infty$ from Proposition \ref{prop::emm}. Hence,~\eqref{eq::Lyapunov} holds for all $x$ with $\emm(x) = \infty$. It is thus sufficient to show that 
\begin{eqnarray}
\label{eq::SimpleLyapunov}
\forall x \in \mathbb{T}(\beta), \quad \min_{u \in \mathbb U} \ \ell(x,u) + \emm(Ax+Bu) \leq \emm(x).
\end{eqnarray}
\noindent
\textit{Part I (Polytopic Control Tube): } 
For given $x \in \mathbb T(\beta)$, let $(y,v,y_{\sf s})$ be an associated minimizer of~\eqref{eq::EMM}. Hence,
\begin{align}
\label{eq:rhs_1}
m(x) = \| x \|_P^2 + \| y - y_{\sf s} \|_{\Gamma}^2 + \sum_{i=1}^{\vp} \| v_i - K V_i y \|_{\Theta}^2.
\end{align}
Proposition~\ref{prop::FCI} guarantees the existence of an input $u \in \mathbb U$ and a nonnegative vector
$\lambda \in \mathbb R^{\vp}$ with $\| \lambda \|_1 = 1$ satisfying
\begin{align}
\label{eq::xuaux}
x &= \sum_{i=1}^{\vp} \lambda_i V_i y, \quad
u = \sum_{i=1}^{\vp} \lambda_i v_i, \\
\label{eq::xpaux}
\text{and} \quad x^+ &:= Ax + Bu \in \Poly(y^+),
\end{align}
where $y^+$ is defined as in Lemma \ref{lem:T_beta}, according to which $(y^+,v^+,y_{\sf s})$ is feasible for~\eqref{eq::EMM} at $x^+ = Ax+Bu$, such that
\begin{align}
\label{eq::11}
\hspace{-3pt}
\scalemath{0.97}{\emm(x^+) \leq \| x^+ \|_P^2 + \| y^+ - y_{\sf s} \|_{\Gamma}^2 + \sum_{i=1}^{\vp} \| v_i^+ - K V_i y^+ \|_{\Theta}^2.}
\end{align}
We now bound the terms on the right-hand-side.
\\ \\ \noindent
\textit{Part II (Difference to LQR Control): }
To bound $\emm(x^+)$, we introduce the difference terms
\begin{align}
\label{eq::edef}
\forall i \in \{ 1, \ldots, {\sf v} \}, \quad
e_i \defeq v_i - K V_i y.
\end{align}
which are deviations of the vertex controls $v_i$ from the LQR controls at the vertices $V_i y$ of the polytope $\Poly(y)$. We define $(y^+,v^+,v_{\sf s})$ as in Lemma~\ref{lem:T_beta}. Now, we have
\begin{align}
e^+_i &:= v_i^+ -K V_i y^+ \notag \\
& \ = K V_i y_{\sf s} + \beta ( v_i-K V_i y_{\sf s}) - K V_i (y_{\sf s} +\beta (y-y_{\sf s})) \notag \\
\label{eq::ep}
& \ = \beta(v_i - K V_i y) = \beta e_i
\end{align}
for all $i \in \{1, \ldots, \vp \}$. Similarly, we find that
\begin{align}
\label{eq::yp}
y^+-y_{\sf s} &= y_{\sf s} + \beta(y-y_{\sf s}) - y_{\sf s} = \beta (y-y_{\sf s}). 
\end{align}
Hence, substituting~\eqref{eq::yp},~\eqref{eq::ep}, and~\eqref{eq::edef} in~\eqref{eq::11}, we have
\begin{align}
\label{eq::12}
\emm(x^+) &\leq \| x^+ \|_P^2 + \beta^2 \| y - y_{\sf s} \|_{\Gamma}^2 + \beta^2 \sum_{i=1}^{\vp} \| e_i \|_{\Theta}^2.
\end{align}
Additionally, we can substitute these relations in~\eqref{eq::xuaux} and~\eqref{eq::xpaux}, define $w := \sum_{i=1}^{\vp} \lambda_i e_i$ and write the input as
\begin{align}
u=Kx + w.
\label{eq::uaux2}
\end{align}
\\ \noindent
\textit{Part III (Properties of the Algebraic Riccati Equation): }
Due to \eqref{eq::uaux2}, we have $x^+=(A+BK)x+Bw$, such that
\begin{align}
\label{eq::first_term}
\| x^+ \|_P^2 &= x^\tr (A+BK)^\tr P (A+BK) x  \\
& \hspace{20pt} + 2x^\tr (A+BK)^\tr P B w + w^\tr B^\tr P B w.\nonumber
\end{align}
By substituting~\eqref{eq::uaux2}, the stage cost $\ell$ in \eqref{eq::stage_cost} satisfies
\begin{align}
\ell( x, u ) &{=} x^\tr ( Q + K^\tr R K + S K + K^\tr S^\tr )x \notag \\
& \hspace{60pt} + 2 x^\tr ( K^\tr R + S ) w + w^\tr R w.
\label{eq:stage_cost_term}
\end{align}
Hence, adding~\eqref{eq::first_term} and~\eqref{eq:stage_cost_term}, we have
\begin{align}
\| x^+ \|_P^2 + \ell(x, u ) &= x^\tr (A+BK)^\tr P (A+BK) x \notag \\
&  + x^\tr \left( Q + K^\tr R K + S K + K^\tr S^\tr \right) x + 
\notag \\
& + 2 x^\tr ( (A+BK)^\tr P B + K^\tr R + S ) w \notag \\
\label{eq::aux232}
& + w( B^\tr P B + R ) w.
\end{align}
After substituting~\eqref{eq::DARE}, the cross-terms cancel out and the expression simplifies to
\begin{align}
\label{eq::SumEquation}
\| x^+ \|_P^2 + \ell(x, u ) = x^\tr P x + w^\tr (B^\tr P B + R) w.
\end{align}
\\ \noindent
\textit{Part IV (Trace Inequality): }
Next, we bound the term $w^\tr (B^\tr P B + R) w$. Since $w = \sum_{i=1}^{\vp} \lambda_i e_i$, we observe that $w^\tr (B^\tr P B + R) w=\lambda^\tr \Delta \lambda,$ where $\Delta$ is given by
\[
\Delta := \left(
\begin{array}{ccc}
e_1^\tr (B^\tr P B + R) e_1 & \ldots & e_{1}^\tr (B^\tr P B + R) e_{\vp} \\
\vdots & \ddots & \vdots \\
e_{\vp}^\tr (B^\tr P B + R) e_{1} & \ldots & e_{\vp}^\tr (B^\tr P B + R) e_{\vp} \\
\end{array}
\right).
\]
Since $\Delta$ is  a symmetric and positive semi-definite matrix, it satisfies a general trace inequality for positive semi-definite matrices~\citep{Shebrawi2013},
\[
\lambda^\tr \Delta \lambda=\mathrm{Tr}(\Delta \lambda \lambda^\tr) \ \leq \ \mathrm{Tr}(\Delta) \mathrm{Tr}(\lambda \lambda^\tr) = \mathrm{Tr}(\Delta) \| \lambda \|_2^2.
\]
Since our assumptions ensure that $\lambda\geq 0$ and $\|\lambda\|_1=1$, we have $\| \lambda \|_2^2\leq 1$.
Hence, it follows that
\begin{align}
\label{eq::Tbound}
w^\tr (B^\tr P B + R) w \leq \mathrm{Tr}(\Delta)
\leq (1-\beta^2) \sum_{i=1}^{\vp} \| e_i \|_{\Theta}^2,
\end{align}
where the last inequality follows from \eqref{eq::Tcond}.
\\ \\ \noindent
\textit{Part V (Lyapunov Descent): }
Collecting the inequalities,
\begin{align}
\label{eq::TightLyapunov}
& \emm(x^+) + \ell(x,u) \\ 
& \hspace{10pt} \leq \| x^+ \|_P^2 + \ell(x,u) + \beta^2 \| y - y_{\sf s} \|_{\Gamma}^2 + \beta^2 \sum_{i=1}^{\vp} \| e_i \|_{\Theta}^2 \nonumber  \\
& \hspace{10pt} \leq \emm(x) - (1-\beta^2)\| y - y_{\sf s} \|_{\Gamma}^2,\nonumber
\end{align}
where the first inequality follows from \eqref{eq::12}, and the second from \eqref{eq::SumEquation}, \eqref{eq::Tbound} and \eqref{eq:rhs_1}. Since $(1-\beta^2)\| y - y_{\sf s} \|_{\Gamma}^2 \geq 0$, we have $\emm(x^+)+\ell(x,u) \leq \emm(x)$ concluding the proof.
%
\end{pf}

\section{Complete formulation}
\label{sec:complete_formulation}
%
%
%
%


Using the proposed terminal region $\mathbb T(\beta)$ and terminal cost $\emm$, we formulate the OCP
\begin{align}
\label{eq:mpc_new}
&\min_{x,u,y,v,y_{\sf s}} \scalemath{0.95}{\sum^{N-1}_{k=0} \ell(x_k,u_k) + \underbrace{
\begin{pmatrix}
    &\hspace{-50pt}  \|x_N\|_P^2+\|y-y_{\sf s}\|_{\Gamma}^2  \vspace{3pt}\\
    &\hspace{50pt} +\underset{i=1}{\overset{\sf v}{\sum}} \Vert v_i - KV_iy \Vert_{\sf {\Theta}}^2
\end{pmatrix}
}_{\text{\sc terminal cost}}} \nonumber \\
&\hphantom{_{x,y}}\text{s.t.}   \left\{ 
\begin{aligned}
&\scalemath{0.95}{\forall k\in\{0,\ldots,N-1\},} \\
& \scalemath{0.95}{x_{k+1} = Ax_{k} + Bu_{k}, \ x_{k+1} \in \mathbb{X}, \ u_{k}\in\mathbb{U},}\\
\smallskip
& \scalemath{0.95}{Fx_{N}\leq y, \ Ey \leq 0, \ E y_{\sf s} \leq 0,  \ \forall i \in \{1,\ldots,{\sf v}\},} \\
&\scalemath{0.95}{F(AV_iy + B U_i v)\leq y_{\sf s} + \beta (y-y_{\sf s}),}\\
&\scalemath{0.95}{F(A + BK) V_i y_{\sf s} \leq y_{\sf s},} \\
&\scalemath{0.95}{V_{i}y\in\mathbb{X}, \ v_i \in \mathbb{U}, \ V_i y_{\sf s} \in \mathbb{X}, \ K V_i y_{\sf s} \in \mathbb U, }
\end{aligned}
\right.
\end{align}
where $x = (x_1,\ldots,x_N)$ and $u = (u_0,\ldots,u_{N-1})$ denote the optimization variables. The initial value $x_0$ is a parameter that is reset to the current state measurement in each iteration. We denote the optimizers of Problem \eqref{eq:mpc_new} as $(x^*(x_0),u^*(x_0),y^*(x_0),v^*(x_0),y_{\sf s}^*(x_0))$, and the associated control law as $\mu_N(x_0)=u_0^*(x_0)$.  
If $\mathbb X$ and $\mathbb U$ are polyhedra,~\eqref{eq:mpc_new} is a convex quadratic programming (QP) problem that can be solved online by using existing MPC tools.

Given that Assumptions~\ref{ass::sets},~\ref{ass::weights},~\ref{ass::LQR}, and~\ref{ass::FEV} hold, the MPC controller~\eqref{eq:mpc} is recursively feasible if the domain of $\emm$ is control invariant and asymptotically stable if $\emm$ satisfies~\eqref{eq::Lyapunov}, see~\citep{Rawlings2009}. As such, the following statements result from Lemma~\ref{lem:T_beta} and Theorem~\ref{thm::clf}.

\begin{enumerate}

\item If $\Gamma$ and $\Theta$ are positive semi-definite, and $\beta \in [0,1]$, then~\eqref{eq:mpc_new} defines a recursively feasible MPC scheme. 

\item For given $\beta \in [0,1)$, if $\Gamma$ and $\Theta$ satisfy \eqref{eq::Tcond}, then~\eqref{eq:mpc_new} defines a recursively feasible and asymptotically stable MPC control scheme on its admissible set.

\end{enumerate}

\subsection{Computational Complexity}
The online optimization problem~\eqref{eq:mpc_new} has $N (n_x + n_u) + 2 \f + \vp n_u$
optimization variables and $N (n_x + n_{\mathbb U} + n_{\mathbb X}) + \f(1+2\vp) + 2 {\sf e} + 2\vp ( n_{\mathbb U} + n_{\mathbb X} )$ constraints, assuming that the given constraint sets $\mathbb X$ and $\mathbb U$ can be represented by $n_{\mathbb X}$ and $n_{\mathbb U}$ convex constraints, respectively. The additional cost from the terminal constraint depends on the number $\f$ of facets, the number $\ed$ of edges, and the number $\vp$ of the cc-polytopes that can be represented by the triple $(F,E,V)$, with the number of variables increasing linearly in $\f$ and $\ve$, and constraints quadratically in $\f \ve$. By designing the feedback gain $K$ such that $x^+=(A+BK)x$ admits low-complexity invariant sets, e.g., \citep{Tahir2013}, this complexity can be reduced.
As the user chooses this triple, we have full control over the computational complexity of the terminal region. Generally, cc-polytopes with larger numbers of facets, edges, and vertices tend to be more flexible, which leads to larger terminal regions. Methods aimed at designing cc-polytopes that balance these attributes can be found in \citep{badalamenti2025efficientconfigurationconstrainedtubempc}.

\section{Numerical Case Studies}
\label{sec:numerics}
\begin{figure}
    \centering
    \includegraphics[width=0.9\linewidth, trim=0.cm 0.cm 0.cm 0.cm, clip]{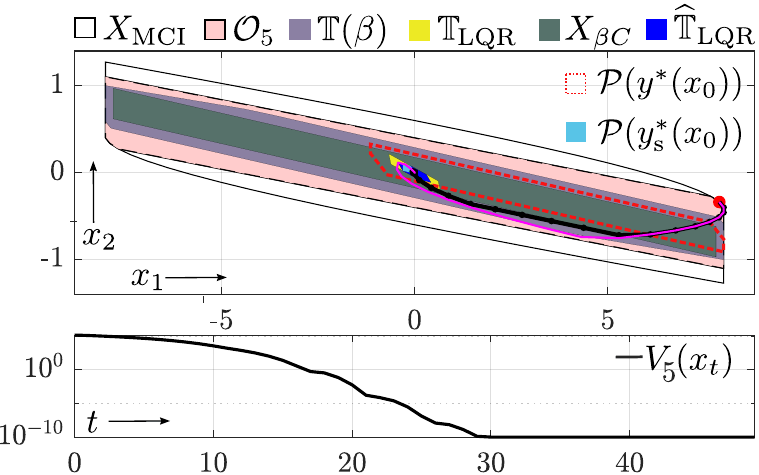}
    \caption{(\textit{Top}) State-space of \eqref{eq:example_1}, along with closed-loop trajectory with $u_t=\mu_5(x_t)$ in black, and infinite horizon solution in magenta, from initial state indicated by the red dot; (\textit{Bottom}) Lyapunov function $V_5(x_t)$.}
    \label{fig:Example_1}
\end{figure}
\begin{figure}
    \centering
    \includegraphics[width=1.0\linewidth, trim=-0.6cm 0.cm 0.cm 0.cm, clip]{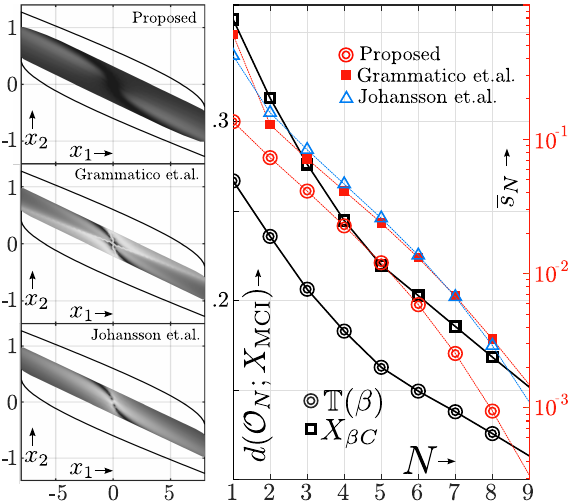}
    \caption{(\textit{Right}) Comparison of admissible region size and suboptimality with benchmark approaches. The black lines refer to axis on the left, with the value denoting the Hausdorff distance between the maximal control invariant set $X_{\mathrm{MCI}}$ and admissible set $\mathcal{O}_N$ of the MPC scheme with horizon length $N$. The axis on the right denotes suboptimality of closed-loop performance of the MPC schemes averaged over $3000$ samples from the admissible region $\mathcal{O}_N$; (\textit{Left}) Distribution of $s_1(x)$ over $x \in \mathcal{O}_1$. Dark colors indicate smaller values.}
    \label{fig:Example_11}
\end{figure}

We present two numerical examples to validate the approach. The first illustrates the method and compares it against state-of-the-art schemes; the second applies it to regulate {\color{black} a stiff LTI model based on the} Klatt–Engel reactor \citep{Klatt1998}. In both cases, we set $\Gamma$ to the identity matrix, and $\Theta=(B^\tr P B+R)/(1-\beta^2)$.
Whenever Problem \eqref{eq:mpc} uses the terminal cost in \eqref{eq::emmLQR} with $\mathbb{T}_{\mathrm{LQR}}$ as the maximal control-invariant set under $u=Kx$, we refer to the resulting scheme as a nominal MPC scheme. For suboptimality analysis, we solve Problem \eqref{eq::ocp} with horizon $N=500$ and terminal constraint $x_{500}=0$, and denote the corresponding trajectory and cost as the optimal infinite-horizon solution and $V_{\infty}(\cdot)$, respectively.
The QP problems are solved using the Gurobi solver in MATLAB R2024b.

\subsection{Illustrative example}
We consider the unstable system
\begin{align}
\label{eq:example_1}
    x^+ = \begin{bmatrix} 1.1 & 2.0 \\ 0 & 0.95 \end{bmatrix} x+ \begin{bmatrix} 0 \\ 0.0787 \end{bmatrix} u
\end{align}
presented in \citep{johansson2024}, subject to constraints $\mathbb{X}=\{x | \|x\|_{\infty} \leq 8 \}$ and $\mathbb{U}=\{u|\|u\|_{\infty}\leq 1\}$. Setting $Q=100I_2$ and $R=1$ yields $K=[-4.6128 \ -18.8646]$. We construct the template matrix $F \in \R^{6 \times 2}$ such that $\Poly(1)$ is the convex hull of the $6$-largest simplicial partitions of the maximal $\lambda$-contractive set \citep{Blanchini2008} with $\lambda=0.95$.
Then, $(F,E,V)$ satisfies Assumption~\ref{ass::FEV}  with $\sf f,\ve,\sf e=6$. Finally, we select $\beta=0.95$ to formulate Problem \eqref{eq:mpc_new}. In Figure~\ref{fig:Example_1}-Top, we plot the maximal control invariant set $X_{\mathrm{MCI}}$, along with the terminal region $\mathbb{T}(\beta)$ in \eqref{eq::Tbeta}. The set $\mathcal{O}_5$ is the admissible set of our MPC scheme with $N=5$.
We also plot the maximal control invariant set $\mathbb{T}_{\mathrm{LQR}}$ with $u=Kx$, and $\widehat{\mathbb{T}}_{\mathrm{LQR}}$ in \eqref{eq::hatTLQR}. Since $\widehat{\mathbb{T}}_{\mathrm{LQR}}$ is the union of all control invariant cc-polytopes with $u=Kx$, $\widehat{\mathbb{T}}_{\mathrm{LQR}} \subseteq \mathbb{T}_{\mathrm{LQR}}$ holds. We plot the closed-loop state trajectory $x_t$ under the MPC law $u_t=\mu_5(x_t)$, initialized at $x_0=(7.8875,-0.3386)$, along with the set $\mathcal{P}(y^*(x_0))$ which satisfies $x^*_5(x_0) \in \mathcal{P}(y^*(x_0))$. This set defines the start of a control invariant tube, which converges to $\mathcal{P}(y_{\sf{s}}^*(x_0)) \subseteq \widehat{\mathbb{T}}_{\mathrm{LQR}}$. Recall from Lemma \ref{lem:T_beta} that this tube sequence satisfies $y_{k+1}= y_{\sf{s}}^*(x_0)+\beta(y_k-y_{\sf{s}}^*(x_0))$, with $y_0=y^*(x_0)$. We also plot the infinite-horizon state trajectory.
In Figure \ref{fig:Example_1}-Bottom, we plot the cost $V_5(x_t)$ of Problem \eqref{eq:mpc_new}, which exhibits a Lyapunov decrease ( Section~\ref{sec:complete_formulation}). To analyze suboptimality, we define
\begin{align} 
\label{eq:suboptimality_measure}
\hspace{-5pt}
s_N(x_0)=\left(\sum_{t=0}^M \ell(x_t,\mu_N(x_t))-V_{\infty}(x_0)\right)/ {V_{\infty}(x_0)}.
\end{align}
For $M=50$, the simulation results in $s_N(x_0)=0.0158$.

\textit{Benchmark Comparisons:}
We compare our approach against \citep{Grammatico2013} and \citep{johansson2024} regarding admissible region size and suboptimality, formulating the benchmarks using the maximal $\beta$-contractive cc-polytope $X_{\beta C}$. Figure \ref{fig:Example_11} (right) shows that using $\mathbb{T}(\beta)$ yields a smaller Hausdorff distance $d(\mathcal{O}_N;X_{\mathrm{MCI}})$ for all $N \in [1,9]$, hence a larger admissible region. Our method also achieves the lowest average suboptimality $\overline{s}_N$ (evaluated over 3000 samples in $\mathcal{O}_N$ with $M=50$); Figure \ref{fig:Example_11} (left) plots the $s_1(x)$ distribution over $\mathcal{O}_1$, where darker colors indicate lower values. While our average QP solution time increases slightly (0.005–0.006 s vs. 0.0035–0.004 s), this trade-off secures a larger region with substantially reduced suboptimality. Even if benchmarks are reformulated with the maximal $\beta$-contractive invariant set to enlarge their admissible regions ($d(\mathcal{O}_1)=0.164$ vs. $0.266$), our suboptimality remains far lower ($\bar{s}_1=0.1354$, compared to $0.7042$ for \citep{Grammatico2013} and $0.5360$ for \citep{johansson2024}). Finally, our template polytope formulation remains practical for systems with high-complexity maximal control invariant sets, where benchmark approaches often become intractable.

\textit{Effect of $\beta$:} We briefly discuss the effect of the $\beta \in [0,1)$, used in the formulation of the terminal set $\mathbb{T}(\beta)$ in \eqref{eq::Tbeta}, and the terminal cost matrix $\Theta$ in \eqref{eq::Tcond}, on the admissible set and suboptimality. As $\beta \to 1$, we have $y_{\sf s}+\beta(y-y_{\sf s}) \to y$, relaxing the contraction requirement and typically enlarging the admissible region.
For the current example, as $\beta \to 1$, the size of the admissible region increases. We obtain $d(\mathcal{O}_1)$ values of $\{6.7116,6.4971,6.0446,4.6539,0.3930,0.1913\}$ for $\beta$ corresponding to $\{0.1,0.3,0.5,0.7,0.9,0.999\}$ respectively. However, since $\mathbb{T}(\beta)$ is nonlinear in $\beta$, this monotonicity is not guaranteed. With regards to suboptimality, small values of $\beta$ result in small $\Theta$. Then, as per \eqref{eq::EMM}, the vertex controls $v_i$ are encouraged to be different from the LQR inputs $KV_iy$. Since this is desirable outside the maximal positive invariant set, the resulting solution might exhibit reduced suboptimality. Instead, as $\beta \to 1$, then $\|v_i-KV_iy\|_{\Theta}^2$ is so strongly penalized that it tends to a constraint, resulting in increased suboptimality.  For the example, from $x_0=(4.5159,-0.7044)$ and $N=5$, we obtain $s_N(x_0)$ values of $\{0.0235,0.0195,0.0253,0.0468,0.0554,0.0570\}$ for $\beta$ values of $\{0.5,0.6,0.7,0.8,0.9,0.999\}$. While suboptimality generally increases with $\beta$, the trend is non-monotonic due to the nonlinear dependence of optimizers on $\beta$.

\subsection{Continuous stirred tank reactor}
\begin{figure}
    \centering
    \includegraphics[width=1\linewidth, trim=0.cm 0.cm 0.cm 0.cm, clip]{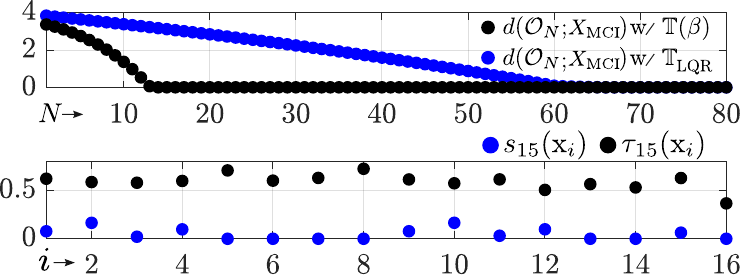}
    \caption{(\textit{Top}) Comparison of admissible region size; (\textit{Bottom}) Comparison of suboptimality and solution time.}
    \label{fig:Example_21}
\end{figure}
\begin{figure}
    \centering
    \includegraphics[width=1\linewidth, trim=0.cm 0.cm 0.cm 0.cm, clip]{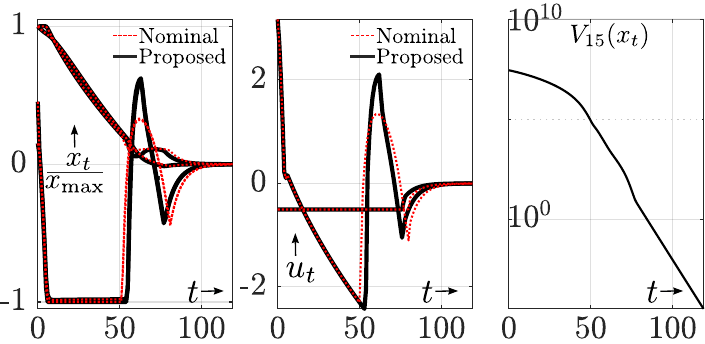}
    \caption{(Left) Normalized closed-loop state trajectories with $N=15$; (Center) Corresponding input trajectories; (Right) Lyapunov function.}
    \label{fig:Example_22}
\end{figure}
We consider a model of the Klatt-Engel reactor presented in \citep{Klatt1998}, which is a jacketed continuous stirred tank reactor with dynamics in \citep[Equation 19]{Klatt1998}, states the concentrations $c_A$ mol/l and $c_B$ mol/l of the reactants, reactor temperature $\nu$ $^\circ$C, and coolant temperature $\nu_K$ $^\circ$C, and the inputs are normalized inflow $u_1$ 1/h and heat removal rate of the coolant $u_2$ MJ/h. We linearize the plant about {\color{black}$(c_A,c_B,\nu,\nu_K)=(0.325,0.321,199.92,194.73)$} and $(u_1,u_2) = (18.83,-4.4957)$, and discretize using the explicit
Runge-Kutta-4 scheme with timestep $0.001$ s, resulting in
the deviation being captured by an LTI system with
{\color{black}\begin{align*}
    [A | B] = 
\scalemath{0.9}{
\left[
\begin{array}{cccc|cc}
  0.2809  & -0.0321  & -0.0098 &  -0.0002  &     0.0124  & -0.0000 \\
    0.0994 &   0.2818  & -0.0081  & -0.0000 & 0.0099  & -0.0000 \\
    1.7800  &  2.7640  &  1.0066  &  0.0295 & -0.0596  &  0.0015 \\
    0.0011  &  0.1629 &   0.0829  &  0.9182 &  -0.0020  &  0.0958
\end{array}
\right]
}.
\end{align*}}
{\color{black} With eigenvalues of $(0.28, 0.34, 0.87, 0.98)$, the system exhibits noticeable stiffness due to a distinct timescale separation: the fast dynamics are driven by the concentration deviations, whereas the slow dynamics dictate the temperature deviations.}
The constraints are $\mathbb{X}=[-x_{\mathrm{max}},x_{\mathrm{max}}]$ with $x_{\mathrm{max}}=(0.05,0.05,4,4)$, and $\mathbb{U}=[-6.83,3.17] \times [-0.5043,0.5057]$. For synthesizing an MPC controller, we first synthesize a feedback gain $K$ along with a matrix $F \in \R^{5 \times 4}$ such that $\mathcal{P}(1)$ is invariant under the control law $u=Kx$ by adapting the procedure in \citep{Mulagaleti2025_PD}. These matrices are
\begin{align*}
   \begin{bmatrix} K \\  \hline F \end{bmatrix} = 
\scalemath{0.75}{
\begin{bmatrix*}[r]
 -36.5225 & -39.3227  & -6.2543  & -0.1072 \\
   -7.3650 & -11.0234  & -3.3091 &  -9.6109 \\  \hline
 -1084.7 & 54 & -30.8 & 170.3 \\
 -322.3 & -1826.6 & -104.6  & 124.6 \\
    -28.0 & -44.1 & -15.1  & -71.1 \\
  14.6 & 21.9 &   6.6  &  19.1 \\
  1376.6 & 1729.3 & 124.3  &  -300.0
\end{bmatrix*}.
}
\end{align*}
We compute the stage cost matrices satisfying Assumptions~\ref{ass::weights} and~\ref{ass::LQR} by solving the inverse LQR problem in \citep[Section III]{Zanon2022}. For this matrix $F$, we compute the matrices $V$ and $E \in \R^{ 1 \times 4}$ such that $(F,E,V)$ satisfies Assumption \ref{ass::FEV}. Using this template, we synthesize Problem \eqref{eq:mpc_new} with $\beta=0.99$.
{\color{black}In Figure \ref{fig:Example_21}-Top, we plot the Hausdorff distance $d(\mathcal{O}_N;X_{\mathrm{MCI}})$.  We report that $d(\mathcal{O}_N;X_{\mathrm{MCI}})=  0.$ for all $N \geq 15$. We compare this distance against that of a nominal MPC scheme, for which $d(\mathcal{O}_N;X_{\mathrm{MCI}})$ converges to $0$ only after $N\geq 62$. }
Hence, our approach results in a much larger admissible region for the same $N$. Equivalently, to achieve the same region of attraction, the nominal MPC scheme would contain $372$ optimization variables and $757$ inequality constraints, while the proposed approach contains $110$ optimization variables and $357$ inequality constraints.
To analyze suboptimality, we project each vertex $\{\mathrm{z}_i, i \in \{1,\cdots,16\}\}$ of $\mathbb{X}$ onto $\mathcal{O}_{15}$ as $\mathrm{x}_i := \arg\min_{x \in \mathcal{O}_{15}} \|x-\mathrm{z}_i\|_2^2.$
Setting the initial state $x_0=\mathrm{x}_i$, we simulate the system with $u=\mu_{15}(x)$ for $M=300$ timesteps. These initial states are also feasible for the nominal MPC scheme for $N= 62$, and its solution matches the approximate infinite horizon solution. In Figure~\ref{fig:Example_21}-Bottom, we plot the suboptimality metric $s_{15}(\mathrm{x}_i)$ over $M=300$ steps. The solution time per iteration of the nominal MPC scheme with $N=62$ is {\color{black}$0.006$s}, and for our proposed scheme with $N=15$ it is {\color{black}$0.0038$s}. Denoting the average solution time from the initial state $x_0=\mathrm{x}_i$ over $M=300$ steps for the nominal MPC scheme by $\tau_{62}^1(\mathrm{x}_i)$ and that of ours by $\tau^2_{15}(\mathrm{x}_i)$, we define $\tau_{15}(\mathrm{x}_i):=(\tau_{62}^1(\mathrm{x}_i)-\tau^2_{15}(\mathrm{x}_i))/\tau_{62}^1(\mathrm{x}_i).$
In Figure~\ref{fig:Example_21}-Bottom, we plot $\tau_{15}(\mathrm{x}_i)$ over the different initial states. We observe that consistently, the suboptimality induced by our scheme is compensated for by the amount of time saved. {\color{black} Over these initial states, we report that our scheme induces on average $5\%$ suboptimality while saving $67.6\%$ computational time,}
thus validating its usage to balance computational complexity and suboptimality. {\color{black} In Figure \ref{fig:Example_22}, we plot the closed-loop state (normalized) and input trajectories from $x_0=(0.0227,0.0075,3.9845,3.9995)$. During the initial steps ($t<50$) when the stage costs dominate, the trajectories are identical to the nominal case. Once the system gets near the terminal set $\mathbb{T}(\beta)$ ($t=56$), the terminal costs dominate, resulting in deviation of the trajectories, and consequently, suboptimality. While this suboptimality can be minimized by increasing the complexity of the cc-polytope, it might lead to increased complexity. We also plot the Lyapunov function, which validates Theorem \ref{thm::clf}.}

\section{Conclusions}
This paper has presented a novel approach for designing terminal ingredients to stabilize MPC schemes. Specifically, a control-invariant terminal set has been constructed using new results from polytopic computation, leading to an enlarged region of attraction, as established in Lemma~\ref{lem:T_beta}. Furthermore, on this set, a piecewise quadratic control Lyapunov function has been constructed, whose values match the optimal LQR cost within a nontrivial neighborhood of the origin, as shown in Theorem~\ref{thm::clf}. 
These ingredients enable the proposed QP-based MPC scheme to achieve a significantly enlarged region of attraction and reduced suboptimality relative to the infinite-horizon optimal control problem, as validated through numerical examples. Future research directions include extending the scheme to economic objectives and uncertain systems.

\bibliographystyle{ifacconf}
\bibliography{ifacconf}

\section{Appendix}
\subsubsection{Proof of Proposition \ref{prop::emm}:}
The first and third statements follow standard result on parametric convex optimization and parametric QPs~\citep{Rockafellar1970,Bemporad2002}. The second statement follows since the constraints in~\eqref{eq::EMM} match the definition of $\mathbb T(\beta)$ in~\eqref{eq::Tbeta}. For the fourth statement, suppose that $(y,v,y_\mathrm{s})$ is a minimizer of~\eqref{eq::EMM} for a given $x \in \widehat{\mathbb T}_{\sf LQR}$. This minimizer exists since $\widehat{\mathbb T}_{\sf LQR} \subseteq \mathbb T(\beta)$ from Lemma~\ref{lem:T_beta}. Next, we define $v_\mathrm{s}$ as in~\eqref{eq::vs}. Since $x \in \widehat{\mathbb T}_{\sf LQR}$, $(y_{\sf s},v_{\sf s},y_{\sf s})$ is feasible for \eqref{eq::EMM}, the inequalities
\begin{align*}
&\hspace{-10pt}\| x \|_P^2 = \Vert x \Vert_{P}^{2} + \Vert y_{\sf s}-y_{\sf s} \Vert_{\Gamma}^2 + \sum^{\sf v}_{i=1}\Vert (v_{\sf s})_i - KV_iy_{\sf s} \Vert_{\Theta}^2 \notag \vspace{-30pt} \\
&\hspace{50pt}\leq \Vert x \Vert_{P}^{2} + \Vert y-y_{\sf s} \Vert_{\Gamma}^2 + \sum^{\sf v}_{i=1}\Vert v_i - KV_iy \Vert_{\Theta}^2, \notag
\end{align*}
follow for any feasible $(y,v,y_{\sf s})$, because $\Gamma,\Theta\succeq 0$. The latter inequality implies that we must have $\emm(x) = \| x \|_P^2$ for all $x \in \widehat{\mathbb T}_{\sf LQR}$, since $(y,v,y_{\sf s})$ is a feasible minimizer. This completes the proof of the proposition.

\end{document}